\documentclass[notitlepage]{article}
\usepackage{amssymb}
\usepackage{amsfonts}
\usepackage{amsmath}
\usepackage{eurosym}
\usepackage{setspace}
\usepackage{graphicx}

\newtheorem{lemma}{Lemma}

\newtheorem{proposition}{Proposition}

\begin{document}

\title{Testing a Graph-Theoretic Condition for Aggregating Incomplete
Rankings: A Technical Note\thanks{%
I am grateful to Tsuyoshi Adachi for useful comments that motivated the
writing of this note. This work was supported by JSPS KAKENHI Grant Number
JP25K05004.}}
\author{Yasunori Okumura\thanks{%
Address: 2-1-6, Etchujima, Koto-ku, Tokyo, 135-8533 Japan.
Phone:+81-3-5245-7300. Fax:+81-3-5245-7300. E-mail:
okuyasu@gs.econ.keio.ac.jp}}
\maketitle

Okumura (2025) studies the problem of aggregating individual rankings when
the submitted rankings may be incomplete. He identifies a key condition,
called Condition 1, for analyzing this problem, but does not address how the
condition can be checked in a computationally efficient way. This note fills
this gap by showing how a standard graph-theoretic algorithm can be used to
test Condition 1. We also use this approach to describe an efficient
implementation of the mechanism $f^{\ast }$ proposed by Okumura (2025).

First, we briefly revisit the model of Okumura (2025).

Let $A$ be the finite set of alternatives. Let $V$ be the finite set of
individuals (voters) who submit a ranking of alternatives. We assume $%
\left\vert A\right\vert \geq 3$ and $\left\vert V\right\vert \geq 3$.

Let $A_{v}\subseteq A$ be a subset of alternatives representing the set of
alternatives that $v$ evaluates and $\mathbf{A}=\left( A_{v}\right) _{v\in
V}\subseteq A^{\left\vert V\right\vert }$. We assume that for any $v\in V$,
at least two alternatives are in $A_{v}$.

An undirected \textbf{graph} is identified with its edge set: it is a set of
unordered pairs of distinct alternatives. The alternatives that appear in
these pairs are called vertices. If $\left\{ a,b\right\} $ is an element of
a graph $G$, then we simply write $ab\in G$ or equivalently $ba\in G$. A
subset of $G$ is called a \textbf{subgraph} of $G$.

A subgraph of $G$ denoted by $P\subseteq G$ is said to be an (undirected) 
\textbf{path} of $G$ if 
\begin{equation*}
P=\left\{ a_{1}a_{2},a_{2}a_{3},\cdots ,a_{M-1}a_{M}\right\} ,
\end{equation*}%
where $a_{1},a_{2},\cdots ,a_{M}$ are distinct. A subgraph of $G$ denoted by 
$C\subseteq G$ is said to be an (undirected) \textbf{cycle} of $G$ if 
\begin{equation}
C=\left\{ a_{1}a_{2},a_{2}a_{3},\cdots ,a_{M-1}a_{M},a_{M}a_{1}\right\} ,
\label{c}
\end{equation}%
where $a_{1},a_{2},\cdots ,a_{M}$ are distinct and $M\geq 3$, which rules
out that $C=\left\{ a_{1}a_{2},a_{2}a_{1}\right\} $ is a cycle.

We construct graphs from a given $\mathbf{A}$. For $v\in V$, let 
\begin{equation*}
G_{v}=\left\{ ab\text{ }\left\vert \text{ }a,b\in A_{v}\text{ and }a\neq
b\right. \right\} ;
\end{equation*}%
meaning that if individual $v$ evaluates both $a$ and $b$, then $ab\in G_{v}$%
. Moreover, 
\begin{equation*}
G=\left\{ ab\text{ }\left\vert \text{ }a,b\in A_{v}\text{ for some }v\in V%
\text{ and }a\neq b\right. \right\} .
\end{equation*}

For a subgraph $G^{\prime }\subseteq G$, let $A\left( G^{\prime }\right) $
be the subset of $A$ such that for all $a\in A\left( G^{\prime }\right) ,$
there is $ab\in G^{\prime }$ for some $b\in A$. A nonempty subgraph $%
G^{\prime }\subseteq G$ is said to be \textbf{connected}\ if for all $a,b\in
A\left( G^{\prime }\right) ,$ there is a path from $a$ to $b$ in $G^{\prime }
$. A vertex $a\in A\left( G^{\prime }\right) $ is said to be an \textbf{%
articulation} \textbf{vertex} of $G^{\prime }$ if there exist distinct $%
b,c\in A\left( G^{\prime }\right) \setminus \{a\}$ such that every path in $%
G^{\prime }$ from $b$ to $c$ contains $a$.

A subgraph $B\subseteq G$ is said to be \textbf{biconnected} if $B$ either
consists of a single edge or is connected and has no articulation vertex.
Moreover, a nonempty subgraph $B\subseteq G$ is said to be a \textbf{%
biconnected component} of $G$ if $B$ is a maximal biconnected subgraph of $G$%
.\footnote{%
Throughout this paper, graphs are identified with their edge sets, and hence
isolated vertices are \textit{not} treated as connected components or
biconnected components.} A biconnected component $B$ is said to be a \textbf{%
clique} if for all $a,b\in A\left( B\right) $ with $a\neq b$, we have $ab\in
B$.

It is well known that, for any graph $G$, all biconnected components of $G$
can be enumerated in linear time by a depth-first search algorithm; see
Hopcroft and Tarjan (1973). Thus, we can let $B_{1},\cdots ,B_{N}$ be the
biconnected components of $G$. Since every single-edge subgraph of $G$ is
biconnected, 
\begin{equation*}
B_{1}\cup \cdots \cup B_{N}=G\text{.}
\end{equation*}

For any biconnected component $B$ of $G$ with $\left\vert A\left( B\right)
\right\vert \geq 3$, and for any distinct $a,b\in A\left( B\right) $, there
exists a cycle in $B$ that contains both $a$ and $b$. Moreover, every cycle
in $G$ is contained in some biconnected component, because every cycle is
biconnected and hence is contained in a maximal biconnected subgraph of $G$.
We have the following result.

\begin{proposition}
Condition 1 is satisfied if and only if for each $n=1,\cdots ,N$, there is
at least one individual who evaluates $A\left( B_{n}\right) ;$ that is,
there is $v_{n}$ such that $B_{n}\subseteq G_{v_{n}}$.
\end{proposition}

\textbf{Proof.} We show the if-part. Suppose for each $n=1\cdots ,N$, there
is $v_{n}$ such that $B_{n}\subseteq G_{v_{n}}$. Each cycle $C$ in $G$ is
contained in some biconnected component $B_{n}$, $C\subseteq G_{v_{n}}$ for
some $n=1\cdots ,N$. Thus, Condition 1 is satisfied.

We show the only-if-part. We have the following result.

\begin{lemma}
If Condition 1 is satisfied, then $B_{n}$ is a clique for all $n=1,\cdots ,N$%
.
\end{lemma}

\textbf{Proof.} Suppose that Condition 1 is satisfied. We arbitrarily fix $%
B_{n}$ and $a,b\in A\left( B_{n}\right) $ with $a\neq b$. If $\left\vert
A\left( B_{n}\right) \right\vert =2$, then $B_{n}=\{ab\}$ is trivially a
clique. Suppose $\left\vert A\left( B_{n}\right) \right\vert \geq 3$. Then,
there is a cycle in $B_{n}$ that contains both $a$ and $b$. Since Condition
1 is satisfied, there is $v$ such that $a,b\in A_{v}$. This implies that $%
ab\in B_{n}$ and therefore $B_{n}$ is a clique. \textbf{Q.E.D.}\newline

Now, we show Proposition 1. First, suppose $\left\vert A\left( B_{n}\right)
\right\vert =2$. In this case, we can let $B_{n}=\{ab\}$. Then, there must
exist an individual $v_{n}$ who evaluates both $a$ and $b$ and thus $%
B_{n}\subseteq G_{v_{n}}$. Second, suppose $\left\vert A\left( B_{n}\right)
\right\vert \geq 3$. By Lemma 1, $B_{n}$ is a clique. Hence there exists a
cycle in $B_{n}$ whose vertex set is exactly $A\left( B_{n}\right) $. By
Condition 1, there is $v_{n}\in V$ such that all alternatives in this cycle
are contained in $A_{v_{n}}$. Thus $B_{n}\subseteq G_{v_{n}}$. \textbf{Q.E.D.%
}\newline

By Proposition 1, Condition 1 can be tested as follows. First, decompose $G$
into its biconnected components $B_{1},\cdots ,B_{N}$. Then, for each
biconnected component $B_{n}$, check whether there exists an individual who
evaluates all alternatives in $A\left( B_{n}\right) $. If such an individual
exists for every biconnected component, then Condition 1 is satisfied;
otherwise, Condition 1 is not satisfied.

Suppose that Condition 1 is satisfied. We can let $v_{1},\cdots ,v_{N}$ be
individuals who evaluate all alternatives $A\left( B_{1}\right) ,\cdots
,A\left( B_{N}\right) $, respectively. Then, $v_{1},\cdots ,v_{N}$ can be
the local dictators in the mechanism $f^{\ast }$ introduced by Okumura
(2025).

To identify $f^{\ast }$, Okumura (2025) introduces the maximal cycles $%
C^{1},\cdots ,C^{X}$ satisfying

\begin{description}
\item[(i)] neither $A\left( C^{x}\right) \setminus A\left( C^{y}\right) $
nor $A\left( C^{y}\right) \setminus A\left( C^{x}\right) $ is empty for all $%
x,y$ such that $x\neq y,$ and

\item[(ii)] if $a$ is not included in any of $C^{1},\cdots ,C^{X}$ ; that
is, if 
\begin{equation*}
a\in A\setminus \left( A\left( C^{1}\right) \cup \cdots \cup A\left(
C^{X}\right) \right) \left( \equiv A^{0}\right) \text{,}
\end{equation*}%
then $a$ is not included in any cycle of $G$.
\end{description}

We distinguish $B_{1},\cdots ,B_{N}$ into two groups. Let 
\begin{eqnarray*}
\mathcal{N}^{2} &\mathcal{=}&\left\{ \left. n=1,\cdots ,N\text{ }\right\vert 
\text{ }\left\vert A\left( B_{n}\right) \right\vert =2\right\} , \\
\mathcal{N}^{\geq 3} &\mathcal{=}&\left\{ \left. n=1,\cdots ,N\text{ }%
\right\vert \text{ }\left\vert A\left( B_{n}\right) \right\vert \geq
3\right\} .
\end{eqnarray*}

Since each $B_{n}$ is a clique, every $B_{n}$ with $n\in \mathcal{N}^{\geq
3} $ contains a cycle whose vertex set is $A\left( B_{n}\right) $. Moreover,
every cycle of $G$ is contained in some biconnected component. Hence the
maximal cycles $C^{1},\cdots ,C^{X}$ correspond exactly to the biconnected
components $B_{n}$ with $n\in \mathcal{N}^{\geq 3}$. Therefore, $\left\vert 
\mathcal{N}^{\geq 3}\right\vert =X$ and, for each $x=1\cdots ,X$, there is a
unique $n_{x}\in \mathcal{N}^{\geq 3}$ such that $A\left( C^{x}\right)
=A\left( B_{n_{x}}\right) $. Since Condition 1 is satisfied, for each $%
n_{x}=1,\cdots ,N,$ $v_{n_{x}}$ satisfies $B_{n_{x}}\subseteq G_{v_{n_{x}}}$%
. This implies each $v_{n_{x}}$ can be a local dictator who evaluates $%
A\left( C^{x}\right) =A\left( B_{n_{x}}\right) $.

On the other hand, for $n\in \mathcal{N}^{2}$, we can let $B_{n}=\left\{
ab\right\} $. Then, $v_{n}$ ($v^{ab}$ in the terminology of Okumura, 2025)
evaluates both $a$ and $b$ and thus it can be a local dictator who evaluates
them.

Therefore, $v_{1},\cdots ,v_{N}$ can be chosen as the local dictators in $%
f^{\ast }$.

\subsubsection*{\textbf{References}}

\begin{description}
\item Hopcroft, J., and Tarjan, R. 1973. Algorithm 447: Efficient Algorithms
for Graph Manipulation. Communications of the ACM, 16(6), 372--378. DOI:
10.1145/362248.362272.

\item Okumura, Y. 2025. Aggregating incomplete rankings. Mathematical Social
Sciences 136, 102423.
\end{description}

\end{document}